\documentclass{amsart}
\usepackage[a4paper,includeheadfoot,margin=0.25in]{geometry}

\usepackage{amsmath}
\usepackage{amsthm}
\usepackage{amssymb}
\usepackage{url}
\newtheorem{theorem}{Theorem}

\def\m#1{\mathsf{#1}} 
\usepackage[ruled,vlined]{algorithm2e}
\DeclareMathOperator{\MinimalPoly}{\bf MinimalPoly}

\newcommand{\cl}[2]{\ensuremath{\mathit{Cl}_{#1,#2}}}

\DeclareMathOperator{\Det}{Det} 
\DeclareMathOperator{\Tr}{Tr}

\DeclareMathOperator{\arcsinh}{arcsinh}

\newcommand{\bbR}{\ensuremath{\mathbb{R}}}
\newcommand{\bbC}{\ensuremath{\mathbb{C}}}

\newcommand{\reverse}[1]{\widetilde{#1}}
\newcommand{\gradeinverse}[1]{\widehat{#1}}

\newcommand{\ii}{\mathrm{i}}

\newcommand{\ee}{\mathrm{e}} 

\def\m#1{\mathsf{#1}}

\def\e#1{\mathbf{e}_{#1}} 

\usepackage{color}

\newcommand{\mycomment}[1]{} 

\begin{document}

\title[The characteristic polynomial in calculation of exponential]
{The characteristic polynomial in calculation of exponential and
elementary functions in Clifford algebras}

\author[1]{Art{\=u}ras Acus${{}^1}$}

\address[1]{Institute of Theoretical Physics and Astronomy, 
University, Saul{\.e}tekio 3, LT-10257 Vilnius, Lithuania}

\author[2]{Adolfas Dargys${{}^2}$}

\address[2]{Center for Physical Sciences and Technology, Semiconductor
Physics Institute, Saul{\.e}tekio 3, LT-10257
Vilnius, Lithuania}

%
%
%
%
%

%







\begin{abstract}
Formulas to calculate multivector exponentials in a base-free representation and in a orthonormal basis
are presented for an arbitrary Clifford geometric algebra
\cl{p}{q}. The formulas are based on the analysis of roots of
characteristic polynomial of a multivector. Elaborate
examples how to use the formulas in practice are presented. The
results may be useful in the quantum circuits or in the problems
of analysis of  evolution of the entangled quantum states.
\end{abstract}

\maketitle


\section{Introduction and notation}\label{sec:notations}
In the Clifford  geometric algebra \cl{p}{q} the exponential
functions of multivector (MV) are frequently
encountered~\cite{Gurlebeck1997,Lounesto1997,Marchuk2020}. It is
enough to mention a famous  half-angle rotor in the geometric
algebra (GA) that finds a wide application from simple rotations
of animated drawings up to applications in the 4-dimensional
relativity theory. Two kinds of exponential  rotors exist which
are related to trigonometric and hyperbolic functions. In higher
dimensional GAs a mixing of trigonometric and hyperbolic functions
was found~\cite{Dargys2022a,AcusDargys2022b}. In this paper the
method to calculate the exponential functions  of a general
multivector {MV} in an arbitrary GA is presented applying for this
purpose the characteristic polynomial in a form of MV. In
Sec.~\ref{sec:charpoly} the methods to generate characteristic
polynomials  in  \cl{p}{q} algebras characterized by arbitrary
signature $\{p,q\}$ and dimension $n=p+q$  are discussed. The
method of calculation of exponential is presented in
Sec.\ref{sec:CLpq}. In Sec.~\ref{sec:otherfunctions} we
demonstrate that the obtained GA exponentials may be used to find
GA elementary and special functions. Below, the notation used in
the paper is described briefly.

In the orthonormalized basis used here the geometric product of
basis vectors $\e{i}$ and $\e{j}$ satisfies the anti-commutation
relation, $\e{i}\e{j}+\e{j}\e{i}=\pm 2\delta_{ij}$. The number of
subscripts indicates the grade. The scalar has no index and is a
grade-0 element, the vectors $\e{i}$ are the grade-1 elements, the
bivectors $\e{ij}$ are grade-2 elements, etc. For mixed signature 
\cl{p}{q} algebra the squares of basis vectors, correspondingly,
are $\e{i}^2=+1$ and $\e{j}^2=-1$, where $i=1,2,\ldots, p$ and
$j=p+1,\ldots, p+q$. The sum
$n=p+q$ is the dimension of the vector space. The general MV
is expressed as
\begin{equation}\label{mvA}
  \m{A}=a_0+\sum_{i}a_i\e{i}+\sum_{i<j}a_{ij}\e{ij}+\sum_{i<j<k}a_{ijk}\e{ijk}+\cdots+\sum_{i<j<\cdots < n=p+q}a_{\underbrace{ij\cdots n}_{p+q=n}} \e{ij\cdots n}
  =a_0+\sum_J^{2^n-1}a_J\e{J},
\end{equation}
where $a_i$, $a_{ij\cdots}$ are the real coefficients. Since the
calculations have been done by GA package written in
\textit{Mathematica} language in the following it appears
convenient  to write the basis elements $\e{ij\cdots}$ and indices
in the reverse degree lexicographic ordering, accordingly. For
example, when $p+q=3$ then the basis elements and indices are
listed in the order $\{1,\e{1},\e{2},\e{3},\e{12},\e{13},
\e{23},\e{123}\equiv I\}$, i.e., the  basis element numbering
always increases from left to right. During computation this
ordering is generated by \textit{Mathematica} program
automatically. This is important because swapping of two adjacent
indices changes basis element sign. The ordered set of indices
will be denoted by single capital letter $J$ referred to as
multi-index the values of which run over already mentioned list
(also see the last expression in~Eq.~\eqref{mvA}). Note, that in
the multi-index representation the scalar is deliberately excluded
from summation as indicated by upper range $2^n-1$  in the sum in
the last expression. The convention is useful since in the
function formulas below the scalar term always has a simpler form.
The highest degree element (pseudoscalar) will be denoted as $I$,
and the corresponding coefficient as $a_I$.

We shall need three grade involutions: the  reversion (e.g.,
$\reverse{\e{12}}=\e{21}=-\e{12}$ ), the inverse (e.g.,
$\gradeinverse{\e{123}}=(-\e{1})(-\e{2})(-\e{3})=-\e{123}$) and
their combination called Clifford conjugation
$\widetilde{\widehat{\e{123}}}=\widehat{\e{321}}=-\e{321}=\e{123}$).
Also we shall need the Hermitian conjugate MV $\m{A}^\dagger$ and
non-zero grade negation operation denoted by overline
$\overline{\m{A}}$. The MV Hermitian conjugation expressed for
basis elements $\e{J}$ in both  real and complex GAs can be
written as~\cite{Marchuk2020,Shirokov2018a}
\begin{equation}\label{hermitian}
  \m{A}^{\dagger}=a_0^*+a_1^*\e{1}^{-1}+\dots+a_{12}^*\e{12}^{-1}
+\dots+a_{123}^*\e{123}^{-1}\dots=\sum_Ja_J^*\e{J}^{-1},
\end{equation}
where $a_J^*$ s the complex conjugated $J$-th
coefficient.\footnote{There is a simple trick to find
$\e{J}^{-1}$. Formally raise all indices and then lower them down
but now taking into account the considered algebra signature
$\{p,q\}$. Finally, apply the reversion. For example, in \cl{0}{3}
we have
$\e{123}\to\e{}^{123}\to-\e{123}\to-\widetilde{\e{123}}\to\e{123}$.
Thus, $\e{123}^{\dagger}=\e{123}$.} For each multi-index that
represents the basis vector with $\e{J}^2=+1$ the Hermitian
conjugation does nothing but changes signs if $\e{J}^2=-1$.
Therefore, the basis elements $\e{J}$ and $\e{J}^\dagger$ can
differ by sign only. The non-zero grade negation is an operation
that changes signs of all grades  to opposite except the scalar,
i.e. $\overline{\m{A}}= a_0-\sum_J^{2^n-1}a_J\e{J}$.

\section{MV characteristic polynomial and equation}
\label{sec:charpoly}

The algorithm to calculate the exponential and associated
functions presented below is based on a characteristic polynomial.
There is a number of methods adapted to MVs, for example, based on
MV determinant, recursive Faddeev-LeVerrier method adapted to GA
and the method related to Bell
polynomials~\cite{Helmstetter2019,Shirokov2021,Abdulkhaev2021}. In
this section the methods  are briefly summarized.

Every MV $\m{A}\in\cl{p}{q}$ has a characteristic polynomial
$\chi_{\m{A}}(\lambda)$ of degree $2^{d}$ in $\bbR$, where
$d=2^{\lceil\tfrac{n}{2}\rceil}$ is the integer, $n=p+q$. In
particular, $d=2^{n/2}$ if $n$ is even and $d=2^{(n+1)/2}$ if $n$
is odd. The integer $d$ may be also interpreted as a dimension of
real or  complex matrix representation of  Clifford algebra in the
8-fold periodicity table~\cite{Lounesto1997}. The characteristic
polynomial that follows from determinant of a MV $\m{A}$ is
defined by
\begin{equation}\label{CharPolyDef}
  \chi_{\m{A}}(\lambda)=-\Det(\lambda -\m{A})=\sum_{k=0}^{d} C_{(d-k)}(\m{A})\, \lambda ^k=
   C_{(0)}(\m{A})\lambda^d+C_{(1)}(\m{A})\lambda^{d-1}+\cdots+C_{(d-1)}(\m{A})\lambda+C_{(d)}(\m{A})\,.
\end{equation}
The variable in the characteristic polynomial  will be always
denoted by $\lambda$ and the roots of $\chi_{\m{A}}(\lambda)=0$
(called the characteristic equation) by $\lambda_i$, respectively.
$C_{(k)}\equiv C_{(k)}(\m{A})$ are the real coefficients defined for the MV $\m{A}$. The coefficient at the highest power
of variable $\lambda$ is always assumed $C_{(0)}=-1$. The
coefficient $C_{(1)}(\m{A})$ is connected with MV trace,
$C_{(1)}(\m{A})=\Tr(\m{A})=d \left\langle\m{A}\right\rangle_{0}$.
The coefficient $C_{(d)}(\m{A})$ is related to MV determinant
$\Det\m{A}=-C_{(d)}(\m{A})$.

Table~\ref{tableDet}   shows  how the MV determinant (which is a
real number) can be calculated in low dimension GAs, $n\le6$. This
table also may be used to find  the coefficients $C_{(k)}(\m{A})$
in the characteristic polynomial~\eqref{CharPolyDef}. For a
concrete algebra it is enough to replace $\m{A}$'s in the
Table~\ref{tableDet} by $(\lambda-\m{A})$. For
example, in case of Hamilton quaternion algebra \cl{0}{2} we have
$\m{A}=a_0+a_1\e{1}+a_2\e{2}+a_{12}\e{12}$ and
$\overline{\m{A}}=a_0-a_1\e{1}-a_2\e{2}-a_{12}\e{12}$, then
\begin{equation}
 \chi_{\m{A}}(\lambda)=-\Det(\lambda
-\m{A})=-(\lambda-\m{A})(\lambda-\overline{\m{A}})=-(a_0^2+a_1^2+a_2^2+a_{12}^2)+2a_0\lambda-\lambda^2.
\end{equation}
Thus, $C_{(0)}=-1$, $C_{(1)}=2 a_0=\Tr\m{A}$ and
$C_{(2)}=-(a_0^2+a_1^2+a_2^2+a_{12}^2)=-\Det\m{A}$  which is in
accord with that calculated from quaternion matrix representation
A=\begin{math}\bigl(\begin{smallmatrix}a_0+\ii a_1&a_2+\ii
a_{12}\\-a_2+\ii a_{12}& a_0-\ii
a_1\end{smallmatrix}\bigr)\end{math}. The table can be also used
to find values of other coefficients $C_{(k)}(\m{A})$ of
polynomial~\eqref{CharPolyDef}. To this end it is enough to
replace $\Det(\m{A})$ in Table~\ref{tableDet} by
$\Det(\lambda-\m{A})$ and then recursively differentiate
 with respect to $\lambda$ a proper number of times,
\begin{equation}
  C_{(k-1)}(\m{A})= -\frac{1}{d - (k - 1)} \left.\frac{\partial C_{(k)}(\lambda-\m{A})}{\partial \lambda}\right|_{\lambda =0}\qquad k=d,\ldots, 1\,,
\end{equation}
which is a straightforward method to obtain the coefficient at
$\lambda^{d-k}$ for any polynomial.

 In Faddeev-Leverrier method and its
modifications~\cite{Helmstetter2019,Householder1975,Hou1998,Shirokov2020a}
the coefficients $C_{(k)}(\m{A})$ in
polynomial~\eqref{CharPolyDef} are calculated recursively,
beginning from $C_{(1)}(\m{A})$ and ending with $C_{(d)}(\m{A})$.
We start from multivector $\m{A}_{(1)}$ by setting
$\m{A}_{(1)}=\m{A}$. Then compute the coefficient
$C_{(k)}(\m{A})=\frac{d}{k}\langle \m{A}_{(k)} \rangle_{0}$ and in
the next step the new MV $\m{A}_{(k+1)}=\m{A}
\bigl(\m{A}_{(k)}-C_{(k)}\bigr)$:
\begin{equation}\label{FLAlg}
  \begin{array}{rcl}
 \m{A}_{(1)}=\m{A}&\rightarrow&C_{(1)}(\m{A})=\frac{d}{1}\langle \m{A}_{(1)}\rangle_0,\\
 \m{A}_{(2)}=\m{A}\bigl(\m{A}_{(1)}-C_{(1)}\bigr)
&\rightarrow&C_{(2)}(\m{A})=\frac{d}{2}\langle \m{A}_{(2)}\rangle_0,\\
    &\vdots&\\
 \m{A}_{(d)}=\m{A}\bigl(\m{A}_{(d-1)}-C_{(d-1)}\bigr)
&\rightarrow&C_{(d)}(\m{A})=\frac{d}{d}\langle
\m{A}_{(d-1)}\rangle_0.
\end{array}
\end{equation}
The determinant of MV then is
$\Det(\m{A})=\m{A}_{(d)}=-C_{(d)}=\m{A}
\bigl(\m{A}_{(d-1)}-C_{(d-1)}\bigr)$. This algorithm if adapted to
GA allows to compute characteristic polynomials for MV of
arbitrary algebra $\cl{p}{q}$. In  alternative recursive
method~\cite{Helmstetter2019} one starts from $C_{(0)}^{\prime}=1$
rather\footnote{This means that all characteristic coefficients
computed with this formula are of opposite sign than obtained by
\eqref{FLAlg}, i.e. $C_{(k)}^{\prime}=-C_{(k)}$. } then
$C_{(0)}(\m{A})=-1$ and initial MV $\m{B}_0=1$, and uses the
following iterative procedure,
 \begin{equation} C_{(k)}^{\prime}=-\text{Tr}(\m{A}\m{B}_{k-1})/k,\qquad
\m{B}_k=\m{A}\m{B}_{k-1}+ C_{(k)}^{\prime},\quad k=1,2,\dotsm, d.
 \end{equation}
The trace may be calculated after multiplication of MVs and taking
the scalar part of the result, $\Tr(\m{A}\m{B}_{k-1})=d \left\langle\m{A}\m{B}_{k-1}\right\rangle_{0}$, or using the trace formula for
products of MVs~\cite{Shirokov2021,Abdulkhaev2021}.
\begin{table}
\begin{center}
  \begin{tabular}{cc}
    $\cl{p}{q}$& $\Det(\m{A})$ \\\hline
\rule{0pt}{3ex}
$p+q=1,2$&$\m{A}\bar{\m{A}}$\\
$p+q=3,4$ &$
    \frac{1}{3}\bigl(\m{A}\m{A} \overline{\m{A}\m{A}}+2 \m{A}\overline{\bar{\m{A}}\overline{\bar{\m{A}}\bar{\m{A}}}} \bigr)$
\\
$p+q=5,6$ &$
    \frac{1}{3}\bigl(\m{H}\m{H} \overline{\m{H}\m{H}}+2 \m{H}\overline{\bar{\m{H}}\overline{\bar{\m{H}}\bar{\m{H}}}} \bigr)\qquad\textrm{with}\quad H=\m{A}\reverse{\m{A}}$
\end{tabular}
\end{center}
  \caption{Optimized expressions for determinant of MV $\m{A}$ in
low dimensional GAs $n=p+q\le 6$. The overbar denotes a negation
of all grades except the scalar,
$\overline{\m{A}}:=\langle\m{A}\rangle_{0}-\sum_{k=1}^n
\langle\m{A}\rangle_{k}=2\langle\m{A}\rangle_{0}-\m{A}$.\label{tableDet}}
\end{table}

Also the coefficients of characteristic polynomial can be deduced
from  complete Bell polynomials. In this approach a set of scalars
is used~\cite{Shirokov2021,Abdulkhaev2021},
\begin{eqnarray}
  S_{(k)}(\m{A}):= (-1)^{k-1}d (k-1)! \langle \m{A}^k \rangle_0,\qquad k=1, \ldots, d,
\end{eqnarray}
where $\langle \m{A}^k \rangle_0$ is the scalar part of MV raised
to  $k$ power.  The needed coefficients are given by
\begin{equation}
C_{(0)}(\m{A})=-1;\qquad
C_{(k)}(\m{A})=\frac{(-1)^{k+1}}{k!}B_k(S_{(1)}(\m{A}),
S_{(2)}(\m{A}), S_{(3)}(\m{A}), \ldots, S_{(k)}(\m{A})),\qquad
k=1, \ldots, d,
\end{equation}
where  $B_k(x_1, \ldots, x_k)$ are the complete Bell polynomials,
The first Bell polynomials\footnote{\textit{Mathematica} v.10 already has function for partial Bell
polynomials BellY[\,]. The Bell  Complete Polynomial then can be
computed as BellCP[x\_\text{List}]:= Sum[BellY[Length[x], k, x],
\{k,1,Length[x]\}], where x\_\text{List} is a list of variables
$x_i$.} are defined by relations
\begin{equation}\begin{split}
B_0&=1,\quad B_1(x_1)=B_0x_1=x_1,\\
B_2(x_1,x_2)&=B_1x_1+B_0x_2=x_1^2+x_2,\\
B_3(x_1,x_2,x_3)&=B_2x_1+\tbinom{2}{1}B_1x_2+B_0x_4=x_1^3+3x_1x_2+x_3,\\
B_4(x_1,x_2,x_3,x_4)&=B_3x_1+\tbinom{3}{1}B_2x_2+\tbinom{3}{2}B_1x_3+B_0x_1=x_1^4+6x_1^2x_2+4
x_1x_3+3x_2^2+x_4,\\
B_5(x_1,x_2,x_3,x_4,x_5)&=B_4x_1+\tbinom{4}{1}B_3x_2+\tbinom{4}{2}B_2x_3+\tbinom{4}{3}B_1x_4+B_0x_5
=x_1^4+6x_1^2x_2+4 x_1x_3+3x_2^2+x_4,
\end{split}\end{equation}
where $\tbinom{n}{r}$ is the binomial coefficient. This sequence
can be easily extended to higher orders. The complete Bell
polynomials also can be represented in a form matrix
determinant~\cite{Wikipedia2022}.

The coefficients of characteristic equation satisfy the following
properties
\begin{align}\label{propertyDC}
  \frac{\partial C_{(k)}(t \m{A})}{\partial t}=k t^{k-1}  C_{(k)}(t \m{A}),\qquad
  \frac{\partial C_{(1)}(t \m{A}^k)}{\partial t}&=k t^{k-1}  C_{(1)}(t \m{A}^k),
\end{align}
where $t$ is a scalar parameter. We will utilize these properties
when proving the exponential formula.

Since provided below formulas contain the sums over roots of
characteristic polynomial, it is worth to remind generalized
Vi\`{e}te's formulas that relate coefficients of characteristic
polynomial to specific sums over the  roots $r_i$:
\begin{align}
  &r_{1}+r_{2}+\cdots+r_{d-1}+r_{d}=(-1)^1\frac{C_{(1)}}{C_{(0)}} \\
  & \left(r_{1} r_{2}+r_{1} r_{3}+\cdots+r_{1} r_{d}\right)+\left(r_{2} r_{3}+r_{2} r_{4}+\cdots+r_{2} r_{d}\right)+\cdots+r_{d-1} r_{d}=(-1)^2\frac{C_{(2)}}{C_{(0)}} \\
&\quad \vdots\notag \\
  &  r_{1} r_{2} \ldots r_{d}=(-1)^{d} \frac{C_{(d)}}{C_{(0)}} .
\end{align}

The other interesting identity~\cite{Hou1998}, which is important for integral
representation of functions is
\begin{equation}
  \Tr{\mathcal{L}\bigl(\ee^{t\m{A}}\bigr)}=\frac{\chi_{\m{A}}^\prime(\lambda)}{\chi_{\m{A}}(\lambda)},
\end{equation}
where $\mathcal{L}$ denotes Laplace transform
$\mathcal{L}\bigl(\ee^{t\m{A}}\bigr)=\bigl(\lambda -\m{A}
\bigr)^{-1}$ of MV $\m{A}$ and $\chi_{\m{A}}^\prime(\lambda)$ is
derivative of the characteristic polynomial
$\chi_{\m{A}}(\lambda)$~(see \eqref{CharPolyDef}) with respect to
polynomial variable $\lambda$.

In matrix theory very important polynomial is a minimal polynomial
$\mu_A(\lambda)$. It establishes the conditions of
diagonalizability of matrix $A$. Similar polynomial
$\mu_{\m{A}}(\lambda)$ may be defined for MV. In particular, it is well-known that
matrix is diagonalizable (aka nondefective) if and only if the
minimal polynomial of the matrix does not have multiple roots, i.e.,
when the minimal polynomial is a product of distinct linear
factors. It is also well-known that the minimal polynomial divides
the characteristic polynomial. If roots of
characteristic equation are all different, then matrix/MV is
diagonalizable. The converse, unfortunately, is not true, i.e.,
MV, the characteristic polynomial of which has multiple roots, may
be diagonalizable.  It is also
established~\cite{Wikipedia2022Diagonalizabitiy} that in case of
matrices over the complex numbers ${\mathbb C}$, almost every
matrix is diagonalizable, i.e., the set of complex $d\times d$
matrices that are not diagonalizable over ${\mathbb C}$ --
considered as a subset of ${\mathbb C}^{d\times d}$-- has the
Lebesgue measure zero. An algorithm how to compute minimal
polynomial of MV without doing reference to matrix representation
of the MV is given in Appendix~\ref{sec:appendix}.

\section{MV exponentials in \cl{\lowercase{p}}{\lowercase{q}} algebra}
\label{sec:CLpq}
\subsection{Exponential of MV in coordinate (orthogonal basis) form}

\begin{theorem}[Exponential in coordinate form]
The exponential of a general $\cl{p}{q}$ MV $\m{A}$ given by Eq.~\eqref{mvA} 
 is the multivector
\begin{align}\label{expNcomplexCoord}
\exp(\m{A})=
  &\frac{1}{d}
  \sum_{i=1}^{d} \exp(\lambda_i)\Biggl(1+
   \sum_{J}^{2^n-1}  \e{J} \,
  \frac{\sum_{m=0}^{d-2} \lambda_i^m \sum_{k=0}^{d-m-2}C_{(k)}(\m{A})\, C_{(1)}(\e{J}^\dagger\m{A}^{d-k-m-1})}{\sum_{k=0}^{d-1}(k+1)\,C_{(d - k-1)}(\m{A})\, \lambda_i^{k}}
 \Biggr)\\=
&\frac{1}{d}
  \sum_{i=1}^{d} \exp\bigl(\lambda_i\bigr)\Bigl(1+\sum_{J}^{2^n-1}  \e{J}
\,b_J(\lambda_i)\Bigr),
 \qquad b_J(\lambda_i)\in\bbR,\bbC\,.
\end{align}
Here $\lambda_i$ and $\lambda_i^j$  denotes, respectively, the
root of a characteristic equation  and the root raised to  power
$j$. The sum is over all roots $\lambda_i$ of characteristic
equation $\chi_{\m{A}}(\lambda)=0$, where $\chi_{\m{A}}(\lambda)$
is the characteristic polynomial  of MV $\m{A}$  expressed as
$\chi_{\m{A}}(\lambda) =\sum_{i=0}^{d} C_{(d-i)}(\m{A})\, \lambda
^i$, and where $C_{(d-i)}(\m{A})$ are the coefficients at variable
$\lambda$ raised to  power~$i$. The symbol
  $C_{(1)}(\e{J}^\dagger\m{A}^{k})=d\, \langle\e{J}^\dagger\m{A}^{k}\rangle_0 $
denotes the first coefficient (the coefficient at $\lambda^{d-1}$)
in the characteristic polynomial that consists of geometric
product of the hermitian conjugate basis element $\e{J}^\dagger$
and $k$-th power of  initial MV:
$\e{J}^\dagger\m{A}^{k}=\e{J}^\dagger\underbrace{\m{A}\m{A}\cdots\m{A}}_{k\
\text{terms}}$.

Note, because the roots of characteristic equation in general are
the complex numbers, the individual terms in sums are complex.
However, the result $\exp(\m{A})$ always simplifies to a real
number (see subsection~\ref{realanswer}).
\end{theorem}

\begin{proof}
The proof is based on checking the following defining equation
below for MV exponential presented in
Eq.~\eqref{expNcomplexCoord}
\begin{equation}
\left.\frac{\partial\exp(\m{A}t)}{\partial t}\right|_{t=1} = \m{A}
\exp(\m{A})=\exp(\m{A}) \m{A}, \end{equation} where the MV $\m{A}$
is assumed to be independent of scalar parameter $t$. After
differentiation with respect to $t$ and then setting $t=1$, one
can verify that the result indeed is $\m{A} \exp(\m{A})$. At this
moment we have explicitly checked the
formula~\eqref{expNcomplexCoord} symbolically when $n=p+q\le 5$
and numerically up to $n\le 10$. To establish non contradicting
nature of~\eqref{expNcomplexCoord} for a general $n$ it seems more
appropriate  to resort to coordinate-free or base-free formula
(see Eq.~\eqref{expNcomplexCoordFree}).
\end{proof}

\vspace{3mm}
 \textbf{Example 1.} {\it Exponential of generic MV in \cl{0}{3} with all different roots.}
Let's compute the  exponential of  $\m{A}=8-6 \e{2}-9 \e{3}+5
\e{12}-5 \e{13}+6 \e{23}-4 \e{123}$ with
Eq.~\eqref{expNcomplexCoord}. We find $d=4$. Computation of
coefficients of characteristic polynomial
$\chi_{\m{A}}(\lambda)=C_{(4)}(\m{A})+C_{(3)}(\m{A})
\lambda+C_{(2)}(\m{A})\lambda^2+C_{(1)}(\m{A})
\lambda^3+C_{(0)}(\m{A})\lambda^4$ for MV $\m{A}$ yields
$C_{(0)}(\m{A})=-1$, $C_{(1)}(\m{A})=32$,  $C_{(2)}(\m{A})=-758$,
$C_{(3)}(\m{A})=10432$, $C_{(4)}(\m{A})=-72693$. The
characteristic equation $\chi_{\m{A}}(\lambda)=0$ then becomes
$-72693+10432 \lambda-758 \lambda^2+32 \lambda^3-\lambda^4=0$,
that has four different roots $\lambda_1=12-\ii \sqrt{53}$,
$\lambda_2=12+\ii \sqrt{53}$, $\lambda_3=4-\ii \sqrt{353},
\lambda_4=4+\ii \sqrt{353}$. For every multi-index $J$ and each
root $\lambda_i$ we have to compute coefficients
\begin{equation*}
\begin{split}
  b_J(\lambda_i)&=\frac{-\lambda_i^2 C_{(1)}(\e{J}^\dagger\m{A})+\lambda_i \bigl(32 C_{(1)}(\e{J}^\dagger\m{A})-C_{(1)}(\e{J}^\dagger\m{A}^2)\bigr)-758 C_{(1)}(\e{J}^\dagger\m{A})+32 C_{(1)}(\e{J}^\dagger\m{A}^2)-C_{(1)}(\e{J}^\dagger\m{A}^3)
  }{-4 \lambda_i^3+96 \lambda_i^2-1516 \lambda_i+10432},
\end{split}
\end{equation*}
where  we still have to substitute the coefficients
$C_{(1)}(\e{J}^\dagger\m{A}^k)$
\begin{equation*}
  \begin{array}{l|rrrrrrrr}
    C_{(1)}(\e{J}^\dagger\m{A}^k)   &\e{1}^\dagger&\e{2}^{\dagger}&\e{3}^{\dagger}&\e{12}^{\dagger}&\e{13}^{\dagger}& \e{23}^{\dagger}&\e{123}^{\dagger}\\[2pt]
\hline
    k=1& 0&-24&-36&20&-20&24&-16\\
    k=2&192&-224&-416&32&-128&384&-856\\
    k=3&8208&5952&5508&-11572&7468&888&-7984
\end{array}\,,
\end{equation*}
different for each multi-index $J$.  The Hermite conjugate
elements are $\e{J}^\dagger=\{-\e{1},-\e{2},-\e{3},
-\e{12},-\e{13}, -\e{23},\e{123}\}$. After substituting all
computed quantities into \eqref{expNcomplexCoord} we finally get,
where  $\alpha=\sqrt{53}$ and $\beta=\sqrt{353}$,
\begin{align}
  \exp(\m{A})= &\frac{1}{2} \ee^4 \left(\ee^8\bigl( \cos(\alpha)+\cos
   \bigl(\beta\bigr)\right) +\left(\frac{3}{\alpha} \ee^{12} \sin \bigl(\alpha\bigr)-\frac{3}{\beta} \ee^4 \sin \bigl(\beta\bigr)\right) \e{1}\notag \\
   &+ \left(\frac{-1}{2 \alpha}\ee^{12} \sin \bigl(\alpha\bigr)-\frac{11}{2 \beta} \ee^4 \sin \bigl(\beta\bigr)\right)
   \e{2}
  +\left(-\frac{2}{\alpha} \ee^{12} \sin
   \bigl(\alpha\bigr)-\frac{7}{\beta} \ee^4 \sin \bigl(\beta\bigr)\right) \e{3}
   \notag\\
  &+ \left(-\frac{2}{\alpha} \ee^{12} \sin \bigl(\alpha\bigr)+\frac{7}{\beta} \ee^4 \sin
   \bigl(\beta\bigr)\right) \e{12}
 +\left(\frac{1}{2 \alpha}\ee^{12} \sin \bigl(\alpha\bigr)-\frac{11}{2
   \beta} \ee^4 \sin \bigl(\beta\bigr)\right) \e{13}
   \\
  &+\left(\frac{3}{\alpha} \ee^{12} \sin \bigl(\alpha\bigr)+\frac{3}{\beta} \ee^4 \sin \bigl(\beta\bigr)\right) \e{23} +\frac{1}{2} \ee^4 \left(\cos \bigl(\beta\bigr)-\ee^8
   \cos \bigl(\alpha\bigr)\right) \e{123}
  .\notag
\end{align}
which (after simplification) coincides with our earlier
result~\cite{AcusDargysPreprint2021}.

\vspace{3mm}
 \textbf{Example 2.} {\it Exponential of generic MV in  \cl{4}{2} with different roots.}
Let's compute the exponential of  $\m{A}=2+3 \e{4}+3
\e{26}+\e{1345}-2 \e{12456}+3 \e{123456}$ using
formula~\eqref{expNcomplexCoord}. In this case $d=8$ and
$\chi_{\m{A}}(\lambda)=C_{(8)}(\m{A})+C_{(7)}(\m{A})
\lambda+C_{(6)}(\m{A}) \lambda^2+C_{(5)}(\m{A})
\lambda^3+C_{(4)}(\m{A}) \lambda^4+C_{(3)}(\m{A})
\lambda^5+C_{(2)}(\m{A})\lambda^6+C_{(1)}(\m{A})
\lambda^7+C_{(0)}(\m{A})\lambda^8$. The coefficients of
characteristic polynomial $\chi_{\m{A}}(\lambda)$ are
$C_{(0)}(\m{A})=-1$, $C_{(1)}(\m{A})=16$,  $C_{(2)}(\m{A})=-64$,
$C_{(3)}(\m{A})=16$, $C_{(4)}(\m{A})=32$, $C_{(5)}(\m{A})=-1280$,
$C_{(6)}(\m{A})=20672$,  $C_{(7)}(\m{A})=-42752$,
$C_{(8)}(\m{A})=14336$. The characteristic equation
$\chi_{\m{A}}(\lambda)=0$ is $14336-42752 \lambda+20672
\lambda^2-1280 \lambda^3+32 \lambda^4+16 \lambda^5-64 \lambda^6+16
\lambda^7-\lambda^8=0$. It has eight different roots
$\lambda_1=-4, \lambda_2=2, \lambda_3=5-\ii \sqrt{3},
\lambda_4=5+\ii \sqrt{3},\lambda_5=-1-\ii \sqrt{15},
\lambda_6=-1+\ii \sqrt{15}, \lambda_7=5- \sqrt{21}, \lambda_8=5+
\sqrt{21}$. Then for every multi-index $J$ and each root
$\lambda_i$ we have  to compute the coefficients
\begin{equation*}
\begin{split}
  &b_J(\lambda_i)=\\\Bigl(
  &C_{0}(A) C_{1}(\e{J}^{\dagger} A^{1}) \lambda _i^6+\bigl(C_{1}(A) C_{1}(\e{J}^{\dagger} A^{1})+C_{0}(A) C_{1}(\e{J}^{\dagger} A^{2})\bigr) \lambda _i^5+\bigl(C_{2}(A) C_{1}(\e{J}^{\dagger} A^{1})+C_{1}(A) C_{1}(\e{J}^{\dagger} A^{2})+C_{0}(A) C_{1}(\e{J}^{\dagger} A^{3})\bigr) \lambda _i^4\\
  &+\bigl(C_{3}(A) C_{1}(\e{J}^{\dagger} A^{1})+C_{2}(A) C_{1}(\e{J}^{\dagger} A^{2})+C_{1}(A) C_{1}(\e{J}^{\dagger} A^{3})+C_{0}(A) C_{1}(\e{J}^{\dagger} A^{4})\bigr) \lambda _i^3
  \\
  &+\bigl(C_{4}(A) C_{1}(\e{J}^{\dagger} A^{1})+C_{3}(A) C_{1}(\e{J}^{\dagger} A^{2})+C_{2}(A) C_{1}(\e{J}^{\dagger} A^{3})+C_{1}(A) C_{1}(\e{J}^{\dagger} A^{4})+C_{0}(A) C_{1}(\e{J}^{\dagger} A^{5})\bigr) \lambda _i^2
  \\
  &+\bigl(C_{5}(A) C_{1}(\e{J}^{\dagger} A^{1})+C_{4}(A) C_{1}(\e{J}^{\dagger} A^{2})+C_{3}(A) C_{1}(\e{J}^{\dagger} A^{3})+C_{2}(A) C_{1}(\e{J}^{\dagger} A^{4})+C_{1}(A) C_{1}(\e{J}^{\dagger} A^{5})+C_{0}(A) C_{1}(\e{J}^{\dagger} A^{6})\bigr) \lambda _i
  \\
  &
   +C_{6}(A) C_{1}(\e{J}^{\dagger} A^{1})+C_{5}(A) C_{1}(\e{J}^{\dagger} A^{2})+C_{4}(A) C_{1}(\e{J}^{\dagger} A^{3})+C_{3}(A) C_{1}(\e{J}^{\dagger} A^{4})+C_{2}(A) C_{1}(\e{J}^{\dagger} A^{5})+C_{1}(A) C_{1}(\e{J}^{\dagger} A^{6})\\
  &\qquad  +C_{0}(A) C_{1}(\e{J}^{\dagger} A^{7})
    \Bigr) \Big/\Bigl(
  8 \lambda_i^7 C_{0}(A)+7 \lambda_i^6 C_{1}(A)+6 \lambda_i^5 C_{2}(A)+5 \lambda_i^4 C_{3}(A)+4 \lambda_i^3 C_{4}(A)+3 \lambda_i^2 C_{5}(A)+2 \lambda_i C_{6}(A)+C_{7}(A)
  \Bigr)\,.
\end{split}
\end{equation*}
 The coefficients $C_{(1)}(\e{J}^\dagger\m{A}^k)$ have values
\begin{equation}
\arraycolsep=3.0pt
\begin{array}{l|rrrrrrrrrrr}
   k&\e{4}^\dagger&\e{15}^{\dagger}&\e{26}^{\dagger}&\e{34}^{\dagger}&\e{145}^{\dagger}& \e{246}^{\dagger}&\e{1256}^{\dagger}&\e{1345}^{\dagger}&\e{2346}^{\dagger}&\e{12456}^{\dagger}&\e{123456}^{\dagger}\\[2pt]
 \hline
1& 24 & 0 & 24 & 0 & 0 & 0 & 0 & 8 & 0 & -16 & 24 \\
2& 96 & 0 & 144 & 0 & -96 & -144 & -96 & -112 & 0 & -64 & 48 \\
3& 1200 & 864 & 1008 & -288 & -672 & -1008 & -576 & -672 & 96 & -960 & 672 \\
4& 9792 & 8064 & 8256 & -1152 & -8832 & -10368 & -8064 & -5312 & -2688 & -7808 & 5568 \\
5& 94848 & 80640 & 82944 & -26496 & -81792 & -91008 & -82560 & -42752 & -24960 & -84992 & 46848 \\
6& 859008 & 787968 & 752256 & -294912 & -826368 & -876672 & -797184 & -397824 & -288768 & -817152 & 370176 \\
7& 8221440 & 7628544 & 7243008 & -3059712 & -7972608 & -8163072 & -7531776 & -3403264 & -3028992 & -8024320 &
   3460608
\end{array}\notag
\end{equation}
In Table, not listed coefficients are zeroes. The Hermitian
conjugate basis elements in the inverse degree lexicographical
ordering are
\begin{equation*}
\begin{split}
&\{\e{1},\e{2},\e{3},\e{4},-\e{5},-\e{6},-\e{12},-\e{13},-\e{14},\e{15},\e{16},-\e{23},-\e{24},\e{25},\e{26},-\e{34},\e{35},\e{36},\e{45},\e{46},-\e{56},-\e{123},-\e{1
   24},\e{125},\e{126},  -\e{134},\\
&\e{135},\e{136},\e{145},\e{146},-\e{156},-\e{234},\e{235},\e{236},\e{245},\e{246},-\e{256},\e{345},\e{346},-\e{356},-\e{456},\e{1234},
   -\e{1235},-\e{1236},-\e{1245},-\e{1246},\e{1256},\\
&
   -\e{1345},-\e{1346},\e{1356},\e{1456},-\e{2345},-\e{2346},\e{2356},\e{2456},\e{3456},-\e{12345},-\e{12346},\e{1235
   6},\e{12456},\e{13456},\e{23456},-\e{123456}\}.
\end{split}
\end{equation*}
Substituting all quantities into \eqref{expNcomplexCoord} after
simplification we get
\begin{align}
  \exp(\m{A})&=\textstyle
\frac{1+\ee^6+2 \ee^3 \cos \sqrt{15}+2 \ee^9
   \left(\cos\sqrt{3}+\cosh \sqrt{21}\right)}{8 \ee^4}
+\frac{-175+175 \ee^6+14 \sqrt{15} \ee^3 \sin \sqrt{15}+10
   \sqrt{3} \ee^9 \left(7 \sin \sqrt{3}+5 \sqrt{7} \sinh \sqrt{21}\right)}{840 \ee^4}\e{4}
\notag   \\
&\textstyle
-\frac{1+\ee^6-2 \ee^3 \cos\sqrt{15}+2 \ee^9 \left(\cos \sqrt{3}-\cosh\sqrt{21}\right)}{8 \ee^4}\e{15}
-\frac{1+\ee^6+2 \ee^3 \cos \sqrt{15}-2 \ee^9 \left(\cos \sqrt{3}+\cosh \sqrt{21}\right)}{8 \ee^4}\e{26}
\notag \\
&\textstyle
+\frac{35-35 \ee^6+14 \sqrt{15} \ee^3 \sin \sqrt{15}+5 \sqrt{3} \ee^9 \left(7 \sin
   \sqrt{3}-\sqrt{7} \sinh \sqrt{21}\right)}{210 \ee^4} \e{34}
+\frac{-175+175 \ee^6-14 \sqrt{15} \ee^3
   \sin \sqrt{15}+10 \sqrt{3} \ee^9 \left(7 \sin \sqrt{3}-5 \sqrt{7} \sinh \sqrt{21}\right)}{840 \ee^4} \e{145}
\notag \\
&\textstyle
+\frac{-175+175 \ee^6+14 \sqrt{15} \ee^3 \sin \sqrt{15}-10 \sqrt{3} \ee^9 \left(7 \sin \sqrt{3}+5 \sqrt{7} \sinh
   \sqrt{21}\right)}{840 \ee^4}\e{246}
-\frac{1+\ee^6-2 \ee^3 \cos \sqrt{15}+2 \ee^9
   \left(\cosh\sqrt{21}-\cos \sqrt{3}\right)}{8 \ee^4}\e{1256}
\notag  \\
&\textstyle
+\frac{-35+35 \ee^6+14 \sqrt{15} \ee^3 \sin \sqrt{15}-5
   \sqrt{3} \ee^9 \left(7 \sin \sqrt{3}+\sqrt{7} \sinh \sqrt{21}\right)}{210 \ee^4}\e{1345}
+\frac{-35+35 \ee^6-14 \sqrt{15} \ee^3 \sin\sqrt{15}+5 \sqrt{3} \ee^9 \left(7 \sin \sqrt{3}-\sqrt{7} \sinh\sqrt{21}\right)}{210
   \ee^4} \e{2346}
\notag  \\
&\textstyle
  +\frac{175-175 \ee^6+14 \sqrt{15} \ee^3 \sin \sqrt{15}+10 \sqrt{3} \ee^9 \left(7 \sin
   \sqrt{3}-5 \sqrt{7} \sinh\sqrt{21}\right)}{840 \ee^4} \e{12456}
  +\frac{-35+35 e^6+14 \sqrt{15} \ee^3 \sin \sqrt{15}+5 \sqrt{3} \ee^9 \bigl(7 \sin\sqrt{3}+\sqrt{7} \sinh
   \sqrt{21}\bigr)}{210 \ee^4} \e{123456}
 \, .\notag
\end{align}
Coefficients at basis elements include both trigonometric and hyperbolic functions.

\subsection{Exponential in basis-free form}

The basis-free exponential follows  from
Eq.~\eqref{expNcomplexCoord} after summation  over the
multi-index~$J$.

\begin{theorem}[MV exponential in basis-free form]\label{theorem2}
In $\cl{p}{q}$ algebra  the exponential of a general MV $\m{A}$
(see Eq.~\eqref{mvA}) can be computed by the following formulas
\begin{align}
\exp(\m{A})=
  & \sum_{i=1}^{d} \exp(\lambda_i)\,\beta(\lambda_i)
  \sum_{m=0}^{d-1} \biggl(\sum_{k=0}^{d-m-1}\lambda_i^k C_{(d-k-m-1)}(\m{A})\biggr)\, \m{A}^{m}\label{expNcomplexCoordFree}
  \\=
  &
  \sum_{i=1}^{d} \exp(\lambda_i)\Biggl(\frac{1}{d}+\beta(\lambda_i)\,
  \sum_{m=0}^{d-2} \biggl(\sum_{k=0}^{d-m-2}\lambda_i^k C_{(d-k-m-2)}(\m{A})\biggr)\, \frac{\bigl(\m{A}^{m+1}-\overline{\m{A}^{m+1}}\bigr)}{2}
 \Biggr)\label{expNcomplexCoordFreeWithDim}\\=
  &  \sum_{i=1}^{d} \exp\bigl(\lambda_i\bigr)\Bigl(\frac{1}{d}+\beta(\lambda_i)\m{B}(\lambda_i)\Bigr),\qquad\mathrm{with}\qquad
  \beta(\lambda_i)=\frac{1}{\sum_{j=0}^{d-1}(j+1)\,C_{(d -j-1)}(\m{A})\, \lambda_i^{j}}\label{beta} \,.
\end{align}
The expression
$\frac{1}{2}\bigl(\m{A}^{m+1}-\overline{\m{A}^{m+1}}\bigr)\equiv
\langle\m{A}^{m+1}\rangle_{-0}$ indicates that all grades of
  multivector $\m{A}^{m+1}$ are included except grade-$0$, since
the scalar part is simply a sum of exponents of eigenvalues divided by $d$.
\end{theorem}
The form of  exponential of MV in Theorem
\eqref{expNcomplexCoordFree} has some similarity with exponential
of square  matrix~\cite{Fujii2012}, where the characteristic
polynomial was used for this purpose too.
\begin{proof}
We will prove the basis-free formula (Theorem~\ref{theorem2}) by
checking the defining equation for the exponential~\eqref{expNcomplexCoordFree},
\begin{equation}\label{definingProperty}
\left.\frac{\partial\exp(\m{A}t)}{\partial t}\right|_{t=1} = \m{A}
\exp(\m{A})=\exp(\m{A}) \m{A}, \end{equation} where $\m{A}$ is
independent of scalar parameter $t$. Since  the MV commutes with
itself, the multiplications of exponential from left and right by
$\m{A}$ gives the same result. Below, after differentiation with
respect to $t$ and then setting $t=1$, we will verify that the
result indeed is $\m{A} \exp(\m{A})$.

  First, using properties of characteristic coefficients \eqref{propertyDC} and
noting that replacement $\m{A}\to\m{A}t$ implies $\lambda_i\to
\lambda_i t$ and performing differentiation
$\left.\frac{\partial\exp(\m{A}t)}{\partial t}\right|_{t=1}$ we
obtain that $\exp(\lambda_i)$  of right hand side of
\eqref{expNcomplexCoordFree} (and also of \eqref{expNcomplexCoord})
is replaced by $\lambda_i\exp(\lambda_i)$,
\begin{align}
  \left.\frac{\partial\exp(\m{A}t)}{\partial t}\right|_{t=1}=&\sum_{i=1}^{d} \lambda_i\exp(\lambda_i)\beta(\lambda_i)\,
  \sum_{m=0}^{d-1} \biggl(\sum_{k=0}^{d-m-1}\lambda_i^k C_{(d-k-m-1)}(\m{A})\biggr)\, \m{A}^{m},
\end{align}
  where the weight factor $\beta(\lambda_i)$ does not play any role in the proof.
Next, we multiply basis-free formula \eqref{expNcomplexCoordFree}
by $\m{A}$
\begin{align}
\m{A}\exp(\m{A})=\sum_{i=1}^{d} \exp(\lambda_i)\,\beta(\lambda_i)
 \sum_{m=0}^{d-1} \biggl(\sum_{k=0}^{d-m-1}\lambda_i^k C_{(d-k-m-1)}(\m{A})\biggr)\, \m{A}^{m+1}\,,
\end{align}
and subtract the second equation from the first for each fixed
root $\lambda_i$, i.e. temporary ignore summation over roots,
\begin{align}\label{difFormula}
  &  \left.\biggl(\left.\frac{\partial\exp(\m{A}t)}{\partial t}\right|_{t=1}- \m{A}\exp(\m{A})\biggr)\right|_{\lambda_i}=
  \exp(\lambda_i)\, \beta(\lambda_i)
  \Bigl(
  \sum_{k=1}^{d} \lambda_i^k C_{(d-k)}(\m{A})-\m{A}^{k} C_{(d-k)}(\m{A})
  \Bigr)\notag\\
  &\qquad  =\exp(\lambda_i) \,\beta(\lambda_i)\Bigl(
  \bigr(\lambda_i^d -\m{A}^{d}\bigr) C_{(0)}(\m{A})+\bigr(\lambda_i^{d-1} -\m{A}^{d-1}\bigr) C_{(1)}(\m{A})+\cdots + \bigr(\lambda_i -\m{A}\bigr) C_{(d-1)}(\m{A})
  \Bigr)\,.
\end{align}
  From Cayley-Hamilton relations (which follow from algorithm~\eqref{FLAlg})
\begin{align*}
  \sum_{k=0}^{d}  \m{A}^k C_{(d-k)}(\m{A})=\m{A}^{d}C_{(0)}(\m{A})+\m{A}^{d-1}C_{(1)}(\m{A})+\cdots + C_{(d)}(\m{A})= &0\notag,\\
  \sum_{k=0}^{d}  \lambda_i^k C_{(d-k)}(\m{A})=  \lambda_i^d C_{(0)}(\m{A})+\lambda_i^{d-1}C_{(1)}(\m{A})+\cdots +
C_{(d)}(\m{A})=&0,
\end{align*}
we solve for the highest powers $\m{A}^{d}$ and $\lambda_i^d$ and
after substituting them into the difference
formula~\eqref{difFormula} after expansion we obtain zero. Since
the identity holds for each of roots $\lambda_i$, it is true for a
sum over roots as well.
\end{proof}

 \textbf{Example 3.} {\it Exponential of MV in \cl{4}{0} with (multiple) zero eigenvalues.}
Let's compute the  exponential of
$\m{A}=-4-\e{1}-\e{2}-\e{3}-\e{4}-2 \sqrt{3} \e{1234}$ with
basis-free formula~\eqref{expNcomplexCoordFreeWithDim}. Using
Table~\ref{tableDet} one can easily verify that $\Det(\m{A})=0$.
For algebra $\cl{4}{0}$ we find $d=4$. The characteristic
polynomial is $\chi_{\m{A}}(\lambda)=C_{(4)}(\m{A})+C_{(3)}(\m{A})
\lambda+C_{(2)}(\m{A})\lambda^2+C_{(1)}(\m{A})
\lambda^3+C_{(0)}(\m{A})\lambda^4=
 -64 \lambda^2-16 \lambda^3-\lambda^4 = -\lambda^2 (8+\lambda)^2$.
The roots are $\lambda_1=0,
\lambda_2=0, \lambda_3=-8, \lambda_4=-8$. Since  multiple roots appear we have to compute
minimal polynomial (see Appendix~\ref{sec:appendix}) of $\m{A}$,
which is $\mu_{\m{A}}(\lambda)=\lambda (8 + \lambda)$. Since
minimal MV has only linear factors, i.e. the polynomial is square
free, the MV is diagonalizable, and the formula \eqref{expNcomplexCoordFreeWithDim} can be applied without modification. It is
also easy to verify that the minimal polynomial divides the
characteristic polynomial
$\chi_{\m{A}}(\lambda)/\mu_{\m{A}}(\lambda)=\frac{-\lambda^2
(8+\lambda)^2}{\lambda (8 + \lambda)}=-\lambda (8 + \lambda)$.
This confirms that non repeating roots of characteristic
polynomial provide sufficient but not necessary criterion of
diagonalizability.
 Then we have
\begin{equation}\label{example3Bi}
\begin{split}
  \beta(\lambda_i)\m{B}(\lambda_i)=& \frac{1}{\sum_{j=0}^{d-1}(j+1)\,C_{(d -j-1)}(\m{A})\, \lambda_i^{j}}\,
  \sum_{m=0}^{d-2} \sum_{k=0}^{d-m-2}\lambda_i^k C_{(d-k-m-2)}(\m{A})\, \langle\m{A}^{m+1}\rangle_{-0}\\
  =&
 \frac{8+\lambda_i}{4 \lambda_i (4+\lambda_i)}\langle\m{A}\rangle_{-0}+ \frac{16+\lambda_i}{4 \lambda_i (4+\lambda_i) (8+\lambda_i) }\langle\m{A}^{2}\rangle_{-0}+\frac{1}{4 \lambda_i (4+\lambda_i) (8+\lambda_i) }\langle\m{A}^{3}\rangle_{-0}
  \\
  =&-\frac{1}{\lambda_i +4}-\frac{1}{4 (\lambda_i +4)}\e{1}-\frac{1}{4 (\lambda_i +4)}\e{2}-\frac{1}{4 (\lambda_i +4)}\e{3}-\frac{1}{4 (\lambda_i +4)}\e{4}-\frac{\sqrt{3}}{2 \lambda_i +8}\e{1234}\/.
\end{split}
\end{equation}
From the middle line one may suppose that sum over roots would yield division by zero
due to zero denominators. The last line, however, demonstrates
that this is not the case, since after collecting terms at basis
elements we see that all potential zeroes in the denominators have
been cancelled. Unfortunately the cancellation does not occur for
non-diagonalizable MVs (see next example). Lastly, after
performing summation $\sum_{i=1}^{d}
\exp(\lambda_i)\bigl(\frac{1}{d}+\beta(\lambda_i)\m{B}(\lambda_i)\bigr)$
over complete set of roots
$\{\lambda_1,\lambda_2,\lambda_3,\lambda_4\}=\{0,0,-8,-8\}$ with
exponent weight factor $\exp(\lambda_i)$ (which can be replaced by
any other function or transformation, see
Sec.~\ref{sec:otherfunctions}) we obtain
\begin{align*}
  \exp(\m{A})= & \frac{1+\ee^8}{2 \ee^8}+
\frac{1-\ee^8}{8 \ee^8}\big(\e{1}+ \e{2}+\e{3}+
   \e{4}-2\sqrt{3}\e{1234}\,\big).
\end{align*}

 \textbf{Example 4.} {\it Exponential of non-diagonalizable
MV in \cl{3}{0}.} Let's find the  exponential of  $\m{A}=-1+2
\e{1}+\e{2}+2 \e{3}-2 \e{12}-2 \e{13}+\e{23}-\e{123}$ with the
help of  base-free formula~\eqref{expNcomplexCoordFreeWithDim}.
For algebra $\cl{3}{0}$ we have $d=4$. The minimal polynomial  is
$\mu_{\m{A}}(\lambda)=-(2+2 \lambda+\lambda^2)^2$ which coincides
with characteristic polynomial
$\chi_{\m{A}}(\lambda)=-\Det(\lambda-\m{A})$ and has multiple
roots $\{-(1+\ii),-(1+\ii),-(1-\ii),-(1-\ii)\}$. Now, if we
proceed as in Example~3 then  for some basis elements we will get
division by zero. To avoid this, we will add a small element  to
exponent, $\m{A} + \varepsilon \e{1}=\m{A}^\prime$, and after
exponentiation and simplification will compute  a limiting
value when $\varepsilon\to 0$. The infinitesimal element
$\varepsilon \e{1}$ may be replaced by any other one which does
not belong to algebra center. We find that
$\chi_{\m{A}^\prime}(\lambda)= -\lambda ^4-4 \lambda ^3+(2
(\varepsilon -4) \varepsilon -8) \lambda ^2+(4 (\varepsilon -6)
\varepsilon -8) \lambda -\varepsilon (\varepsilon ((\varepsilon
-8) \varepsilon +20)+8)-4$, the limit of which is
$\lim_{\varepsilon\to 0 }
\chi_{\m{A}^\prime}(\lambda)=\chi_{\m{A}}(\lambda)$. If
$\varepsilon$ is included it has four (now different) roots
$\lambda_1=-(1+i)-\sqrt{\varepsilon ^2-(4+2 i) \varepsilon }$,
$\lambda_2=-(1+i)+\sqrt{\varepsilon ^2-(4+2 i) \varepsilon }$,
$\lambda_3=-(1-i)-\sqrt{\varepsilon ^2-(4-2 i) \varepsilon }$,
$\lambda_4=-(1-i)+\sqrt{\varepsilon ^2-(4-2 i) \varepsilon }$
which in the limit $\varepsilon\to 0$ return back to  multiple
roots. Since the roots with $\varepsilon$ included  are
different in calculation of $\beta(\lambda_i)\m{B}(\lambda_i)$ the
division by zero disappears,
\begin{equation*}
\begin{split}
  \beta(\lambda_i)\m{B}(\lambda_i)= &\frac{ \left(-2
   \varepsilon ^2-8 \varepsilon +\lambda_i ^2+4 \lambda_i +8\right)\langle\m{A}^{\prime}\rangle_{-0}}{4 \left(-\varepsilon ^2 \lambda_i -\varepsilon
   ^2-4 \varepsilon  \lambda_i -6 \varepsilon +\lambda_i ^3+3 \lambda_i ^2+4 \lambda_i +2\right)}
    +\frac{ (\lambda_i +4)\langle\m{A}^{\prime 2}\rangle_{-0}}{4 \left(-\varepsilon ^2 \lambda_i -\varepsilon ^2-4
   \varepsilon  \lambda_i -6 \varepsilon +\lambda_i ^3+3 \lambda_i ^2+4 \lambda_i +2\right)}\\
  &\qquad+\frac{\langle\m{A}^{\prime 3}\rangle_{-0}}{4 \left(-\varepsilon ^2 \lambda_i -\varepsilon ^2-4 \varepsilon  \lambda_i -6 \varepsilon +\lambda_i ^3+3
   \lambda_i ^2+4 \lambda_i +2\right)}\\
  =&\frac{1}{4}\biggl(1+\frac{1}{\lambda_i ^3+3 \lambda_i ^2+\bigl(4-\varepsilon  (\varepsilon +4)\bigr) \lambda_i + 2 -\varepsilon (\varepsilon +6)}\biggl(
 \bigl(
(\varepsilon +2) \lambda_i ^2+2 (\varepsilon +3) \lambda_i-\varepsilon  (10+\varepsilon
   (\varepsilon +6))+2\bigr) \e{1}\\
&\qquad\qquad  +\bigl(\lambda_i^2+6\lambda_i-\varepsilon  (\varepsilon +8)+4\bigr) \e{2}  + 2\bigl(\lambda_i ^2-\varepsilon  (\varepsilon +2)-2\bigr) \e{3}
  + 2\bigl(-\lambda_i^2-4\lambda_i+\varepsilon (\varepsilon +6)-2\bigr) \e{12}\\
&\qquad\qquad +2\bigl(-\lambda_i^2 -\lambda_i +\varepsilon  (\varepsilon +3)+1\bigr) \e{13}
+ \bigl(\lambda_i ^2-2 (\varepsilon +1) \lambda_i +(\varepsilon -2) \varepsilon -4\bigr) \e{23}\\
  &\qquad\qquad- \bigl(\lambda_i ^2-2 (\varepsilon -1)\lambda_i +\varepsilon  (\varepsilon +2)+2\bigr) \e{123}
  \biggr)\biggr)\,.
\end{split}
\end{equation*}
After summation over all roots $\{\lambda_1,\lambda_2,
\lambda_3,\lambda_4\}$  in $\sum_{i=1}^{4}
\exp\bigl(\lambda_i\bigr)\Bigl(\frac{1}{4}+\m{B}(\lambda_i)\Bigr)$,
we collect terms at basis elements and finally compute the limit
$\varepsilon\to 0$ for each of coefficients. Then, after
simplification we get the following answer,
\begin{equation*}
\begin{split}
\exp(\m{A})= &  \frac{1}{\ee}\bigl(\cos (1) +(\sin (1)+2 \cos (1))\e{1} + (2 \sin (1)+\cos (1))\e{2}+2  (\cos (1)-\sin (1))\e{3}
    -2 (\sin (1)+\cos (1))\e{12}\\
 &\qquad+ (\sin (1)-2 \cos (1))\e{13}+ (\cos (1)-2 \sin (1))\e{23}-\sin (1) \e{123}\bigr)\,.
\end{split}
 \end{equation*}

It should be noted that  computation of the limit is highly
nontrivial task, especially when dealing with the roots of high
degree polynomial equations. The primary purpose of Example~4 was
to show that non-diagonalizable matrices/MVs represent some
limiting case and the (symbolic) formula is able to  take into
account this case. To illustrate how complicated computation of
exponential of non-diagonalizable matrix for higher dimensional
Clifford algebras could be we have tested internal
\textit{Mathematica} command MatrixExp[\,] using $\cl{4}{2}$
algebra and non-diagonalizable MV $\m{A}^{\prime\prime}=
-1-\e{3}+\e{6}-\e{12}-\e{13}+\e{15}-\e{24}-\e{25}+\e{26}-\e{34}-\e{35}+\e{36}-\e{45}+\e{56}+\e{123}+\e{124}+\e{126}+\e{134}+\e{135}+\e{136}+\e{146}+\e{234}-\e{235}-\e{236}-\e{245}-\e{246}-\e{256}+\e{456}
-\e{1236}+\e{1245}-\e{1246}+\e{1256}-\e{1345}-\e{1346}-\e{1356}+\e{1456}-\e{2346}-\e{2356}+\e{2456}+\e{3456}+\e{12345}-\e{12346}+\e{12356}
$ that was converted to $8\times 8$ real matrix
 representation. The respective MV has minimal polynomial
$(\lambda -1)^2 \left(\lambda ^6+10 \lambda ^5+39 \lambda ^4+124
\lambda ^3+543 \lambda ^2-198 \lambda -4743\right)$.
\textit{Mathematica} (version 13.0) has crashed after almost $48$
hours  of computation after  all 96~GB of RAM was exhausted. This
strongly contrasts with the exponentiation of diagonalizable
matrix of the same \cl{4}{2} algebra: it took only a fraction of a
second to complete the task.

\subsection{Making the answer real}
\label{realanswer}
Formulas~\eqref{expNcomplexCoord} and~\eqref{expNcomplexCoordFree}
include summation over (in general complex valued) roots of
characteristic polynomial, therefore, formally  the result is a
complex number. Here  we are dealing with real Clifford algebras,
consequently, a pure imaginary part or numbers in the final result
must vanish. Because the characteristic polynomial is made up of
real coefficients, the roots of the polynomial always come in
complex conjugate pairs. Thus, after summation over each of
complex root pair in  exponentials ~\eqref{expNcomplexCoord}
and~\eqref{expNcomplexCoordFree} (and other real valued) functions
one must get a real final answer. Indeed, assuming that symbols
$a,b,c,d,g,h$ are real and computing the sum over a single complex
conjugate root pair we come to the following relation
\begin{align*}
  \exp(a + \ii b) \frac{c + \ii d}{g + \ii h} + \exp(a - \ii b) \frac{c - \ii d}{g - \ii h}=
  \frac{2 \ee^{a} \bigl((c g+d h) \cos b+(c h-d g) \sin
b\bigr)}{g^2+h^2},
\end{align*}
the right hand side of which now formally represents a real number as
expected. The left hand side is exactly the expression which we
have in~\eqref{expNcomplexCoord} and~\eqref{expNcomplexCoordFree}
formulas after substitution of pair of complex conjugate roots.
However, from symbolic computation point of view the issue is not
so simple. In general, the roots of high degree ($d\ge 5$)
polynomial equations cannot be solved in radicals and, therefore,
in symbolic packages they are represented as enumerated formal
functions/algorithms of some irreducible polynomials.  In
\textit{Mathematica} the formal solution is represented as
Root[poly,\,k]. In order to obtain a real answer, we have to know
how to manipulate  these formal objects algebraically. To
that end there exist algorithms which allow to rewrite the
coefficients of irreducible polynomials 'poly' after they have
been algebraically manipulated. The operation, however, appears to
be nontrivial and time consuming. In \textit{Mathematica} it is
implemented by RootReduce[\,] command, which produces another
Root[poly$^\prime$,\,k$^\prime$] object. Such a  root reduction
typically raises the order of the  polynomial. From pure numerical
point of view, of course, we may safely  remove spurious complex
part in the final answer to get real numerical value.

\section{Elementary functions of MV}
\label{sec:otherfunctions} The exponential
formulas~\eqref{expNcomplexCoord} and~\eqref{expNcomplexCoordFree}
are more universal than we have expected. In fact they allow to
compute any function and transformation of MV (at least for
diagonalizable MV) if one replaces the exponential weight
$\exp(\lambda_i)$ by any other function (and allow to use complex
numbers). Here we shall demonstrate how to compute $\log (\m{A}),
\sinh (\m{A})$, $\arcsinh (\m{A})$ and Bessel $J_0(\m{A})$ GA
functions of MV of argument $\m{A}$  in $\cl{4}{0}$
(Example~3). This example with zero and negative
eigenvalues was chosen  to demonstrate that no problems arise if
symbolic manipulations are addressed.

After replacement of  $\exp(\lambda_i)$ by  $\log(\lambda_i)$
in~\eqref{expNcomplexCoordFree} and summing up over all roots one
obtains
\begin{equation}\label{logfromexp}
\begin{split}
  \log{\m{A}}=\frac{1}{2}(\log (0_{+})
  +\log (-8))+\frac{1}{8} (\log (-8)-\log (0_{+})) \left(\e{1}+\e{2}+\e{3}+\e{4}+2 \sqrt{3}\,\e{1234}\right).
\end{split}
 \end{equation}
We shall not attempt  to explain what $\log (-8)$ means in
\cl{4}{0} since we want to avoid presence of complex numbers in
real $\cl{4}{0}$. If we assume, however, that $\exp\bigl(\log
(-8)\bigr)=-8$ and $\exp\bigl(\log (0_{+})\bigr)=\lim_{x\to
0_{+}}\exp\bigl(\log (x)\bigr)=0$. Then it is easy to check that
under these assumptions the  exponentiation of $\log{\m{A}}$
yields $\exp(\log(\m{A}))=\m{A}$, i.e., the $\log$ function in
Eq.~\eqref{logfromexp} is formal inverse of $\exp$.

There are no problems when computing hyperbolic and trigonometric
functions and their inverses\footnote{It looks as if the complex
numbers are inevitable in computing  trigonometric functions in
most of real Clifford algebras, except $\cl{3}{0}$ as well as few
others~\cite{Chappell2014}.}. Indeed, after replacing
$\exp(\lambda_i)$ by  $\sinh(\lambda_i)$, and $\arcsinh(\lambda_i)$
in~\eqref{expNcomplexCoordFree} one finds, respectively,
\begin{equation}
\begin{split}
  \sinh{\m{A}}=&\frac{1}{8}\sinh(8)\bigl(-4-\e{1}- \e{2}-\e{3}-\e{4}-2 \sqrt{3}\,\e{1234}\bigr),\\
  \arcsinh{\m{A}}=&\frac{1}{8}\arcsinh(8)\bigl(-4 -\e{1} - \e{2}-\e{3}-\e{4} -2 \sqrt{3}\, \e{1234} \bigr),\\
  J_0(\m{A})=& \frac{1}{2} (1+J_0(8))+
 \frac{1}{8} (J_0(8)-1) \bigl(\e{1}+\e{2}+\e{3}+\e{4}+2\sqrt{3}\,\e{1234} \bigr) \,.
\end{split}
\end{equation}
In the last line $\exp(\lambda_i)$ was replaced by Bessel
$J_0(\m{A})$ function. It is easy to check that
$\sinh\bigl(\arcsinh(\m{A})\bigr)=\m{A}$ is indeed satisfied. We
do not question where special functions of MV/matrix
argument may find application. The purpose of last
computation was to demonstrate that the
formulas~\eqref{expNcomplexCoord} and~\eqref{expNcomplexCoordFree}
are more universal: they allow to compute much larger class
of functions and transformations of MV, because the sum operator
in the formulas satisfies the properties,
\begin{equation}
  \sum_{i=1}^{d}\beta(\lambda_i)\m{B}(\lambda_i)=0,\qquad \sum_{i=1}^{d} \lambda_i\Bigl(\frac{1}{d}+\beta(\lambda_i)\m{B}(\lambda_i)\Bigr)=\m{A}\,,
\end{equation}
where the scalar $\beta(\lambda_i)$ is given in formula
\eqref{beta}. These expressions  provide interesting spectral
decomposition of MV.

\section{Conclusion}
\label{sec:conlusion} The paper shows that  in Clifford geometric
algebras the exponential of a general multivector is associated
with the characteristic polynomial of the multivector (exponent)
and may be expressed in terms of roots of respective
characteristic equation.  In higher dimensional algebras the
coefficients at basis elements, in agreement
with~\cite{AcusDargys2022b}, include a mixture of trigonometric
and hyperbolic functions. The presented GA exponential
formulas~\eqref{expNcomplexCoord} and~\eqref{expNcomplexCoordFree}
can be generalized to trigonometric, hyperbolic and other
elementary functions as well. Besides the explicit examples of functions provided in the article, we were able to compute fractional powers of MV, many special functions available in Mathematica, in particular, HermiteH[\,],LaguerreL[\,] (also with rational parameters), and some of hypergeometric functions.

\section{Appendix: Minimal polynomial of MV}
\label{sec:appendix}

A simple algorithm for computation of matrix minimal polynomial is
given in \cite{Mathworld2022}. It starts by constructing $d\times
d$ matrix $M$ and its powers $\{1, M, M^2,\ldots\}$ and then
converting each of matrix into vector of length $d\times d$. The
algorithm then checks consequently  the sublists $\{1\}$, $\{1,
M\}$, $\{1, M, M^2\}$ etc until the vectors in a particular
sublist are detected to be linearly dependent. Once  linear
dependence is established the algorithm returns polynomial
equation, where coefficients of linear combination are multiplied
by corresponding variable $\lambda^i$.

In GA, the orthonormal basis elements $\e{J}$ are linearly
independent, therefore it is enough to construct vectors made from
 coefficients of MV. Then, the algorithm  starts searching
when these coefficient vectors become linearly dependent.

A vector constructed of matrix representation of MV has
$d^2=\bigl(2^{\lceil\tfrac{n}{2}\rceil}\bigr)^2$ components. This
is exactly the number of coefficients ($2^n$) in MV for Clifford
algebras of even $n$ and twice less than number of matrix elements
$d\times d$ for odd $n$. The latter can be easily
understood if one remembers that for odd $n$ the matrix
representation of Clifford algebra has block diagonal form.
Therefore only a single block will suffice for
required matrix algorithm. Below the
Algorithm~\ref{minimalPoly} describes how to compute minimal
polynomial of MV without employing  a  matrix representation.
\begin{algorithm}
\SetAlgoLined \SetNoFillComment \LinesNotNumbered
\SetKwInput{KwInput}{Input}
\SetKwInput{KwOutput}{Output}
\SetKwProg{minimalPoly}{minimalPoly}{}{}$\MinimalPoly{(\m{A})}${\\
\KwInput{multivector $\m{A}=a_0+\sum_J^{2^n-1}a_J\e{J}$ and polynomial variable $x$}
\KwOutput{minimal polynomial $c_1+c_2 x+c_3 x^2+\cdots$}
  {\scriptsize \tcc{Initialization}}
nullSpace=\{\};\quad lastProduct=1;\quad vectorList=\{\}\;
  {\scriptsize\tcc{keep adding new MV coefficient vectors to vectorList until null space becomes nontrivial}}
  While[nullSpace===\{\},\\
  lastProduct=$\m{A}\circ$lastProduct\;
  AppendTo[vectorList,\, ToCoefficientList[lastProduct]]\;
  nullSpace=NullSpace[Transpose[vectorList]];
  ]\\
  {\scriptsize\tcc{use null space weights to construct the polynomial $c_1+c_2 \m{A}+c_3 \m{A}^2+\cdots$, with $\m{A}$ replaced by given variable $x$}}\vskip 5pt
  \Return{$\mathrm{First[nullSpace]}\cdot \{x^0,x^1,x^2,\ldots,x^{\mathrm{Length[nullSpace]-1}}\}$}.
} \label{minimalPoly}\caption{Algorithm for finding minimal polynomial of
MV in $\cl{p}{q}$}
\end{algorithm}

All functions in the above code are internal {\it Mathematica}
functions, except symbol $\circ$ (geometric product) and
ToCoefficientList[\,] which  is rather simple. It takes MV
$\m{A}$ coefficients and outputs vector coefficients, i.e.
ToCoefficientList[$a_0+a_1\e{1}+a_2\e{2}+\cdots+a_I I$]$\to
\{a_0,a_1,a_2,\ldots,a_I \}$. The real job is done by {\it
Mathematica} function NullSpace[\,], which searches for linear
dependency of inserted vector list. This function is standard
function of linear algebra library. If the list of vectors is
found to be linearly dependent it outputs weight factors of
linear combination for which sum of vectors become zero, or
an empty list otherwise. The AppendTo[vectorList, newVector]
appends the newVector to list of already checked vectors in
vectorList.




Supporting information available as part of the online article:
\url{https://github.com/ArturasAcus/GeometricAlgebra}.

\bibliographystyle{REPORT}

\bibliography{expND}%

%
%

\end{document}